\documentclass[10pt]{bmc_article}

\usepackage{cite} 
\usepackage{url}  
\usepackage{ifthen}  
\usepackage{multicol}   
\usepackage[utf8]{inputenc} 

\usepackage{graphicx}
\usepackage{amsmath,epsfig}
\usepackage{color}
\usepackage{amssymb}
\usepackage{stfloats}
\usepackage{algorithm}
\usepackage{algorithmic}
\usepackage{epstopdf}
\usepackage{amsthm}

\newtheorem{theorem}{\textbf{Theorem}}

\newtheorem{corollary}{\textbf{Corollary}}

\makeatletter

\newcommand{\Rmnum}[1]{\expandafter\@slowromancap\romannumeral #1@}
\makeatother

\usepackage{amsfonts}

\def\includegraphics{}

\newboolean{publ}

\newenvironment{bmcformat}{\baselineskip20pt\sloppy\setboolean{publ}{false}}{\baselineskip20pt\sloppy}

\begin{document}
\begin{bmcformat}

\title{An Effective Handover Analysis for the Randomly Distributed Heterogeneous Cellular Networks}

\author{Bin Fang,
         \email{Bin Fang - fcb096@mail.ustc.edu.cn}
         Wuyang Zhou\correspondingauthor
         \email{Wuyang Zhou\correspondingauthor - wyzhou@ustc.edu.cn}
      }

\address{%
    Wireless Information Network Lab, University of Science and Technology of China, Hefei, China, 230026
}%

\maketitle

\begin{abstract}
Handover rate is one of the most import metrics to instruct mobility management and resource management in wireless cellular networks. In the literature, the mathematical expression of handover rate has been derived for homogeneous cellular network by both regular hexagon coverage model and stochastic geometry model, but there has not been any reliable result for heterogeneous cellular networks (HCNs). Recently, stochastic geometry modeling has been shown to model well the real deployment of HCNs and has been extensively used to analyze HCNs. In this paper, we give an effective handover analysis for HCNs by stochastic geometry modeling, derive the mathematical expression of handover rate by employing an infinitesimal method for a generalized multi-tier scenario, discuss the result by deriving some meaningful corollaries, and validate the analysis by computer simulation  with multiple walking models. By our analysis, we find that in HCNs the handover rate is related to many factors like the base stations' densities and transmitting powers, user's velocity distribution, bias factor, pass loss factor and etc. Although our analysis focuses on the scenario of multi-tier HCNs, the analytical framework can be easily extended for more complex scenarios, and may shed some light for future study.
\end{abstract}
\section*{Keywords}
Stochastic geometry modeling, handover rates, heterogeneous cellular network.
\newpage
\ifthenelse{\boolean{publ}}{\begin{multicols}{2}}{}

\section{Introduction}
With the dramatically increasing of wireless traffic as well as the population of wireless terminals, the traditional homogeneous cellular network cannot provide sufficient bandwidth for all the wireless terminals. In response to the capacity challenges, smaller coverage base stations (BSs) are deployed in hotspots to offload and have a range of tens of meters to several hundred meters. This brings heterogeneity to traditional cellular network and gives birth to the heterogeneous cellular network (HCN). Heterogeneity is expected to be a key feature of the next generation of cellular networks, and an essential means for providing higher network capacity as well as expanded indoor and cell-edge coverage. In general, HCNs comprise a conventional cellular network overlaid with a diverse set of lower-power BSs such as micro cells, femtocells and perhaps relay BSs.

The BSs in different tiers of HCNs (the tiers of BSs are ordered by transmit power) may share the same spectra and have different coverage. Handover happens when a user leaves the coverage of its serving BS and handover rate is defined as the number of handovers per unit time. The handover in HCNs can be divided into two types: horizontal handover and vertical handover. Horizontal handover is the handover between two BSs in the same tier and vertical handover is the handover between two BSs in different tiers. Compared with horizontal handovers, vertical handovers are more difficult to implement because the HCNs may be deployed by different service providers. Thus, to implement vertical handovers, extra communication overhand between the the HCNs is essential. Moreover, vertical handovers would lead to additional transfer delay, jitter and high risk of dropping, which degrade the service quality. Therefore, handover rate especially vertical handover rate is one of the most metrics to instruct the deployment of mobility management and resource management.


The mathematical expression of handover rate in the homogeneous cellular network has been derived by the regular hexagon coverage model \cite{cite_last01} and stochastic geometry model \cite{cite_new03}\cite{cite_new04}. But for the heterogeneous cellular network, there has not been any reliable and generalized handover expression, due to the randomness of the BS positions of HCNs and the different transmitting power of different tiers. In the real deployment, the HCN BSs are distributed irregularly, sometimes in an anywhere plug-and-play manner, which results in a high level of spatial position randomness. BSs in different tiers have different transmitting power for communication, leading to different cell size for different tiers. As a consequence, it is difficult to characterize the cell boundaries and to track boundary crossings by UEs (i.e. handovers) in the global networks.
Few previous works have resolved the above challenges and give a reliable handover rate expression.

In the recent years, stochastic geometry modeling has shown its admirable ability of modeling the position distribution of HCN BSs \cite{cite_06}\cite{cite_07}, and has provided tractable accurate performance bounds for cellular wireless networks. Stochastic geometry is a very powerful mathematical and statistical tool for the modeling, analysis and design of HCNs with random topologies.  For instance, the modeling is employed for the capacity analysis of random channel access schemes like ALOHA \cite{cite_04} and carrier sensing multiple access (CSMA) \cite{cite_05}, the capacity analysis of single and multiple tier cellular networks \cite{cite_last02} and the capacity analysis of cognitive-based networks \cite{cite_last03}.

Hence, in this paper, we investigate the handover rates including horizontal and vertical handover rates by the stochastic geometry modeling. By employing an infinitesimal method, the mathematical expression of instantaneous handover rate is derived for a typical moving UE. From the derivation, we find that the instantaneous handover rate is related to the instantaneous moving speed, and is independent of the moving direction. That means only the moving speed distribution of the memoryless walking model contributes to the handover rates and the handover rates can be derived through averaging the instantaneous handover rates by the moving speed distribution. Thus, the derived handover rate expressions are applicable for all the memoryless walking models. The derived expressions are validated by computer simulation with multiple walking models and the impacts of system parameters like BS density, transmit power, moving velocity of UE, path loss factor are evaluated. Although our analysis focuses on the scenario of multi-tier HCNs, the analytical framework can be easily extended for more complex scenarios.

\section{Downlink System Model}

A fairly general model of HCNs considered in this paper contains $N$ tiers of BSs that are distinguished by their spatial densities, transmit powers, path loss exponents and biasing factors. For instance, as shown in Figure 1, high-power macrocell BS networks are overlaid with successively denser and lower power picocells and femtocells. Macrocell BSs and femtocell BSs can be well modeled by spatial random processes \cite{cite_06}\cite{cite_new02}. Under this model, the positions of BSs in the \emph{n}th-tier are modeled according to a homogeneous PPP (Poisson point process) $\Phi_n$ with intensity $\lambda_n$ in an Euclidean plane.

Each BS in the \emph{n}th-tier has the same transmit power $P_n$, and has the same path loss exponents $\alpha_n>2$, ${n=1,...,N}$. Assume that UEs are uniformly distributed in the Euclidean plane with density of $f_u$, and the movements of UEs are memoryless and are independent of the distributions of BSs. Memoryless here means the current position of a UE is only related to its latest position and is independent of its more earlier position, i.e. $\mathbb{P}[S(t_0)|S(t_1),S(t_2),\cdots]=\mathbb{P}[S(t_0)|S(t_1)]$, where $\mathbb{P}[x]$ is the probability of x, $S(t_i)$ is the position of a UE at time $t_i$, and $t_i > t_{i+1}$, i = 0,1,2,$\cdots$.

We assume open access which means a user is allowed to access any tier's BSs. And consider a cell association based on maximum biased-received-power (BRP) (termed biased association), where a mobile UE is associated with the strongest BS in terms of long-term averaged BRP at the UE. The BRP from the \emph{j}th BS in the \emph{n}th-tier is  $P_{r,nj}$ that can be given by
\begin{equation}
\label{equ_new01th}
P_{r,nj} = P_nL_0(R_{nj}/r_0)^{-\alpha_n}B_n
\end{equation}
where $R_{nj}$ is the distance of the \emph{j}th BSs in the \emph{n}th-tier from the origin, $L_0$ is the path loss at the reference distance $r_0$ (typically about $(4\pi/\nu)^{-2}$ for $r_0=1$, where $\nu$ denotes the wavelength). And $B_n$ is the bias factor of admission \cite{cite_new05}, that could extends the cell range (or coverage) of the \emph{n}th-tier by employing $B_n>1$. The considered BRP is a long-term averaged value and fading is averaged out, and so does not include fading.
We assume that handover happens only when the UE is going across the boundary of the current BS's coverage, which is determined by the long-term averaged BRP and is shown as Figure 1.

\section{Problem Formulation}

Consider a typical UE that is at the origin and is admitted to the \emph{k}th BS of the \emph{m}th-tier initially. As the typical UE moves, it may immigrate to other BSs in the same tier or other tiers. So its admission state can be depicted as Figure 2, that at time \emph{t}, the typical UE is admitted to the \emph{k}th BS with probability $P_{a,k}(t)$, or is admitted to other BS with probability $P_{a,k}(0)-P_{a,k}(t)$. Thus, the instantaneous transition rate from k state to the $\overline{k}$ state at time \emph{t} is
\begin{equation}
\label{equ_01th}
H_k^m(t) = -\frac{dP_{a,k}(t)}{dt}
\end{equation}

Since the movement of the typical UE is a memoryless process and is independent of BSs' distributions, then the instantaneous transition rate is stable and can be given by
\begin{equation}
\label{equ_02th}
H_k^m \triangleq  -\lim\limits_{t\rightarrow 0}\frac{dP_{a,k}(t)}{dt}
\end{equation}

$H_k^m$ is the instantaneous handover rate of the typical UE indeed. Then, the handover rate in a region with area $S$ can be given by
\begin{equation}
\label{equ_02.1th}
\lambda_h = \mathbb{E}[\sum\limits_{m=1}^NH_k^mf_uS]=\sum\limits_{m=1}^N(\mathbb{E}[H_k^m])f_uS
\end{equation}
where $\mathbb{E}[x]$ is the expectation of variable $x$, and $\mathbb{E}[H_k^m]$ is the average handover rate of the typical UE. As the handover is assumed to happen at the boundary of BS coverage, not at the boundary of a specified region with area S, the average handover rate is independent of the shape of the specified region.
In the following section, we would derive the arithmetic expression of $\mathbb{E}[H_k^m]$.

\section{Derivation of Handover Rate}

In this section, we would derive $\mathbb{E}[H_k^m]$ in two steps: firstly, we use an infinitesimal method to derive the instantaneous handover rate at time 0 of the typical UE with instantaneous velocity $v$, then average the instantaneous handover rate by the distribution of the velocity $v$. Thus, the impacts of walking model on the handover rate can be reflected by the distribution of the velocity $v$, i.e. the walking model decides the distribution of velocity $v$, then further affects the average handover rate. So the analysis is applicable for all the memoryless walking models with that thought. Note that the moving direction of the typical UE would not affect the instantaneous handover rate and we would give the explanation in the following derivation.

In the following, we would give the position distribution of the associated BS firstly, and then derive the handover probability of $(P_{a,k}(0)-P_{a,k}(t))$, the  instantaneous handover rate ($H_k^m$) and the average handover rate ($\mathbb{E}[H_k^m]$) in turn.

\subsection{Position Distribution of Associated BS}

Denote $(R_{nj},\theta_{nj})$ as the polar coordinate of the \emph{j}th BS in the \emph{n}th-tier. Assume that the typical UE is admitted to the \emph{k}th BS of the \emph{m}th-tier initially, thus $P_{r,mk} > P_{r,nj}$ for all $n \in \{1,\cdots,N\}$. According to the max-BRP based association and the BRP definition in equation (\ref{equ_new01th}), the distance boundary condition of these unassociated BSs can be derived as equation (\ref{equ_04th}) based on $P_{r,mk} > P_{r,nj}$,
\begin{equation}
\label{equ_04th}
R_{nj} > (\frac{P_nB_n}{P_mB_m})^{\frac{1}{\alpha_n}}(R_{mk})^{\frac{\alpha_m}{\alpha_n}}\triangleq R_{n}^{lb}
\end{equation}
where $R_{n}^{lb}$ is defined as the distance lower bound of the \emph{n}th-tier BSs for the clarity of expression.

According to the distance lower bound of each tier, the probability density function (PDF) of the associated BS's distance ($R_{mk}$) can be derived by using the null probability of a 2-D Poisson process with density $\lambda$ in an area A, which is $\exp(-\lambda A)$. By setting $\lambda=\lambda_n$ and $A=\pi (R_n^{lb})^2$ for the \emph{n}th tier, (n=1,$\cdots$,N)\cite{cite_new05}, we could give the probability density function (PDF) of $R_{mk}$ as
\begin{equation}
\label{equ_new02th}
f(R_{mk})=2\pi\lambda_mR_{mk}\exp\{-\pi\sum\limits_{n=1}^{N}\lambda_n(R_{n}^{lb})^2\}
\end{equation}
where $R_m^{lb}=R_{mk}$. Note that $\int\limits_{0}^{\infty}f(R_{mk})dR_{mk}\triangleq \gamma_m$ is the probability that the typical UE is admitted to a \emph{m}th-tier BS.

Since BSs are deployed as a PPP (Poisson point process), the $\theta_{mk}$ is uniformly distributed in the range of $[0,2\pi]$, and its PDF is given by
\begin{equation}
\label{equ_new03th}
f(\theta_{mk}) = \frac{1}{2\pi}
\end{equation}

Since the distributions of $\theta_{mk}$ and $R_{mk}$ are independent of the coordinate axis and the moving direction of UE, we can assume that the X axis is along the moving direction of the typical UE at time 0.

\subsection{Derivation of Handover Probability}

When the typical UE moves an infinitesimal distance of $r$, i.e. the typical UE moves to the point $(r,0)$ with $r\rightarrow 0$, the BRP from the \emph{j}th BS in the \emph{n}th-tier is
\begin{equation}
\label{equ_05th}
P_{r,nj}^{new} \!\!=\frac{P_nB_nL_0}{(\sqrt{(R_{nj}\cos(\theta_{nj})\!-\!r)^2\!+\!(R_{nj}\sin(\theta_{nj}))^2})^{\alpha_n}}
\end{equation}

According to the max-BRP based association, after the typical UE moves to the new point, it is still admitted to the primary BS \emph{k} only when the BRP from BS \emph{k} is larger than the BRP from anyone else. Given the position of BS \emph{k}, the probability that the typical UE keeps the primary link to the BS k is denoted by $P_{a,k}(r|R_{mk},\theta_{mk})$ and is given by equation (\ref{equ_06th}).
\begin{equation}
\label{equ_06th}
\begin{array}{l}
P_{a,k}(r|R_{mk},\theta_{mk}) = \prod\limits_{n=1}^N\mathbb{P}[P_{r,mk}^{new} \geq P_{r,nj}^{new}]
= \prod\limits_{n=1}^N\mathbb{P}[\cos(\theta_{nj})\leq \frac{R_{nj}^2+x_{nj}}{2rR_j}]
\end{array}
\end{equation}
where $\mathbb{P}[x]$ denotes the probability of event x, and $x_{nj}$ is defined as

\begin{equation}
\label{equ_new04th}
x_{nj}\!=\!r^2\!-\!(\frac{B_nP_n}{B_mP_m})^{\frac{2}{\alpha_n}}(R_{mk}^2\!-\!2rR_{mk}\cos(\theta_{mk})\!+\!r^2)^{\frac{\alpha_m}{\alpha_n}}
\end{equation}
And $\cos(\theta_{nj})\leq \frac{R_{nj}^2+x_{nj}}{2rR_j}$ is derived according to $P_{r,mk}^{new} \geq P_{r,nj}^{new}$ in equation (\ref{equ_06th}).

Hence, according to equation (\ref{equ_06th}), the typical UE would keep its primary link if all the BSs in the \emph{n}th-tier (n=1,$\cdots$,N) are in the region of $\cos(\theta_{nj})\leq \frac{R_{nj}^2+x_{nj}}{2rR_j}$, or wound not other wise.

We call the region of $\cos(\theta_{nj}) > \frac{R_{nj}^2+x_{nj}}{2rR_j}$ as the bad region of \emph{n}th-tier BSs, which is shown as Figure 3. It means that if there are \emph{n}th-tier BSs in the \emph{n}th-tier bad region, the typical UE would immigrate from the serving BS to one of those BSs. Thus, $P_{a,k}(r|R_k,\theta_k)$ equals to the probability that no BS is in its bad region for all tiers.


Denote the area of the \emph{n}th-tier bad region as $A_{mn}(r,R_{mk},\theta_{mk})$, then we can give the null probability of the the PPP $\Phi_n$ in the \emph{n}th-tier bad region as $\exp(\lambda_n A_{mn}(r,R_{mk},\theta_{mk}))$. Since all the PPPs $\{\Phi_n\}_{n=1,\cdots,N}$ are independent, $P_{a,k}(r|R_k,\theta_k)$ is the product of those null probabilities and can be given by the following equation
\begin{equation}
\label{equ_07th}
P_{a,k}(r|R_k,\theta_k)=\prod\limits_{n=1}^N\exp(-\lambda_n A_{mn}(r,R_{mk},\theta_{mk}))
\end{equation}

According to $\cos(\cdot)\leq 1$ and the definition of bad region, the $\theta_{nj}$ boundary conditions of the \emph{n}th-tier bad region can be given as
\begin{equation}
\label{equ_08th}
1 \geq \cos(\theta_{nj})> \frac{R_{nj}^2+x_{nj}}{2rR_j},
\end{equation}

Thus, based on the $R_{nj}$ boundary condition in equation (\ref{equ_04th}) and $\theta_{nj}$ boundary conditions in equation (\ref{equ_08th}), we can further derive the boundary conditions of the \emph{n}th-tier bad region in the Appendix \ref{appendix_01}, and give the results as equation (\ref{equ_09th}), where $\vartheta_{nj}$ is defined as $\vartheta_{nj} = \theta_{nj}^{max}-\theta_{nj}^{min}$, $\theta_{nj}^{max}$ and $\theta_{nj}^{min}$ are the upper bound and lower bound of $\theta_{nj}$, respectively. The shapes of bad regions with different boundary conditions in equation (\ref{equ_09th}) can be depicted as Figure 3.

As Figure 3 (a) shows, when $R_n^{lb}<R_{mk}$ holds, according to the derivation in Appendix \ref{appendix_01}, the range of $R_{nj}$ is $[R_n^{lb},r+\sqrt{r^2-x_{nj}}]$ and the range of $\theta_{nj}$ is $[-\arccos(\frac{R_{nj}^2+x_{nj}}{2rR_{nj}}),\arccos(\frac{R_{nj}^2+x_{nj}}{2rR_{nj}})]$ for a certain $R_{nj}$.

As Figure 3 (b) shows, when both $R_n^{lb}> R_{mk}$ and $|\theta_{mk}|<\arccos(\frac{-R_{mk}}{R_n^{lb}})$ hold, the ranges of $R_{nj}$ and $\theta_{nj}$ are the same as the ranges of $R_{nj}$ and $\theta_{nj}$ in Figure 3 (a). When both $R_n^{lb}> R_{mk}$ and $|\theta_{mk}|>\arccos(\frac{-R_{mk}}{R_n^{lb}})$ hold, as shown by Figure 3 (c), according to the analysis in Appendix \ref{appendix_01}, the range of $R_{nj}$ is $[R_n^{lb},r+\sqrt{r^2-x_{nj}}]$, and the range of $\theta_{nj}$ is $[-\pi,\pi]$ if $R_{nj}\in [R_n^{lb},-r+\sqrt{r^2-x_{nj}}]$, or is $[-\arccos(\frac{R_{nj}^2+x_{nj}}{2rR_{nj}}),\arccos(\frac{R_{nj}^2+x_{nj}}{2rR_{nj}})]$ if $R_{nj}\in [-r+\sqrt{r^2-x_{nj}},r+\sqrt{r^2-x_{nj}}]$.

\begin{equation}
\scriptsize
\label{equ_09th}
\left\{\begin{array}{lc}
\!\!\!R_{nj}\in \emptyset, \vartheta_{nj}=0,&\!\!\!\!\!\!\!\!\!\!\!\!\!\!\!\cos(\theta_{mk})>\frac{R_{mk}}{R_n^{lb}}\\
\!\!\!R_{nj}\in [R_n^{lb},r+\sqrt{r^2-x_{nj}}],\vartheta_{nj}=\left\{\begin{array}{lc}
\!\!\!2\pi,&\!\!\!R_{nj}\in [R_n^{lb},-r+\sqrt{r^2-x_{nj}})\\
\!\!\!2\arccos\frac{R_{nj}^2+x_{nj}}{2rR_{nj}},&\!\!\!R_{nj}\in [-r + \sqrt{r^2-x_{nj}},r+\sqrt{r^2-x_{nj}}]
\end{array}
\right.
& \!\!\!\!\!\!\!\!\!\!\!\!\!\!\!\cos(\theta_{mk})<-\frac{R_{mk}}{R_n^{lb}}\\
\!\!\!R_{nj}\in (R_n^{lb},r+\sqrt{r^2-x_{nj}}), \vartheta_{nj}=2\arccos(\frac{R_{nj}^2+x_{nj}}{2rR_{nj}}), & \!\!\!\!\!\!\!\!\!\!\!\!\!\!\! -\frac{R_{mk}}{R_n^{lb}}\leq\cos(\theta_{mk})\leq\frac{R_{mk}}{R_n^{lb}}
\end{array}
\right.
\end{equation}

Based on the boundary conditions, the area of the \emph{n}th-tier bad region ($A_{mn}(r,R_{mk},\theta_{mk})$) can be derived as the equation (\ref{equ_10th}), where $I(c)$ is the index function that equals 1 if the condition c holds or 0 otherwise. The first term in equation (\ref{equ_10th}) corresponds to the case of $\frac{-R_{mk}}{R_n^{lb}}<\cos(\theta_{mk})<\frac{R_{mk}}{R_n^{lb}}$, which is depicted by Figure 3 (a) and (b), and the second and third terms correspond to the case of $\cos(\theta_{mk})<\frac{-R_{mk}}{R_n^{lb}}$ that is depicted by Figure 3 (c).

\begin{equation}
\label{equ_10th}
\begin{array}{l}
A_{mn}(r,R_{mk},\theta_{mk})=[\int\limits_{R_n^{lb}}^{r+\sqrt{r^2-x_{nj}}}2\arccos(\frac{R_{nj}^2+x_{nj}}{2rR_{nj}})R_{nj}dR_{nj}]I(-\frac{R_{mk}}{R_n^{lb}}\leq\cos(\theta_{mk})\leq\frac{R_{mk}}{R_n^{lb}})\\
+[\int\limits_{-r+\sqrt{r^2-x_{nj}}}^{r+\sqrt{r^2-x_{nj}}}2\arccos(\frac{R_{nj}^2+x_{nj}}{2rR_{nj}})R_{nj}dR_{nj} + \int\limits_{R_n^{lb}}^{-r+\sqrt{r^2-x_{nj}}}2\pi R_{nj}dR_{nj} ]I(\cos(\theta_{mk})<-\frac{R_{mk}}{R_n^{lb}})
\end{array}
\end{equation}

Thus, the handover probability can be obtained as $(P_{a,k}(0)-\mathbb{E}_{\{R_{mk},\theta_{mk}\}}P_{a,k}(r|R_{mk},\theta_{mk}))$.

\subsection{Handover Rate Derivation}

Since the instantaneous handover rate $H_k^m$ is the derivative of handover probability, according to equations (\ref{equ_02th}) and (\ref{equ_07th}), the instantaneous handover rate $H_k^m$ can be derived as equation (\ref{equ_11th}).
\begin{equation}
\label{equ_11th}
\begin{array}{l}
H_k^m = -\lim\limits_{t\rightarrow 0}\frac{dP_{a,k}(t)}{dt}=-\lim\limits_{t\rightarrow 0}\frac{d\mathbb{E}_{\{R_{mk},\theta_{mk}\}}[P_{a,k}(r|R_{mk},\theta_{mk})]}{dt}\\
\!\!=\! -\!\mathbb{E}_{\{R_{mk},\theta_{mk}\}}[\lim\limits_{t\rightarrow 0}\frac{dP_{a,k}(r|R_{mk},\theta_{mk})}{dt}]\\
\!\!=\!\!-\mathbb{E}_{\{R_{mk},\theta_{mk}\}}[\lim\limits_{r\rightarrow 0}\!(\!\frac{dP_{a,k}(r|R_{mk},\theta_{mk})}{dr}\frac{dr}{dt})]\\
\!\!\mathop=\limits^{(a)}\!\!-\mathbb{E}_{\{R_{mk},\theta_{mk}\}}[\lim\limits_{r\rightarrow 0}\frac{d}{dr}(-\sum\limits_{n=1}^N\lambda_nA_{mn}(r,R_{mk},\theta_{mk})
\cdot \exp(-\sum\limits_{n=1}^N\lambda_nA_{mn}(r,R_{mk},\theta_{mk})))]\cdot v \\
\!\!\mathop =\limits^{(b)}\!\mathbb{E}_{\{R_{mk},\theta_{mk}\}}[\sum\limits_{n=1}^N(\lambda_n\cdot \lim\limits_{r\rightarrow 0}\frac{dA_{mn}(r,R_{mk},\theta_{mk})}{dr})] \cdot v
\end{array}
\end{equation}
where $v=\lim\limits_{r\rightarrow 0}\frac{dr}{dt}$ is the instantaneous velocity of the typical UE at time $t=0$, (a) is obtained according to the expression of $P_{a,k}(r|R_{mk},\theta_{mk})$ in equation (\ref{equ_07th}), and (b) is obtained due to $\lim\limits_{r\rightarrow 0}A_{mn}(r,R_{mk},\theta_{mk})=0$ and $\lim\limits_{r\rightarrow 0}\exp(-\sum\limits_{n=1}^N\lambda_nA_{mn}(r,R_{mk},\theta_{mk}))=1$.

Denote $H_k^{m-n}$ as the instantaneous handover rate from the \emph{k}th BS in the \emph{m}th-tier to the BSs in the \emph{n}th-tier.
Thus
\begin{equation}
\label{equ_new06th}
H_k^m = \sum\limits_{n=1}^N H_k^{m-n}
\end{equation}

According to the derivation of $H_k^m$ and the expression of $H_k^m$ in equation (\ref{equ_11th}), we can give $H_k^{m-n}$ as
\begin{equation}
\label{equ_new05th}
H_k^{m-n} = \mathbb{E}_{\{R_{mk},\theta_{mk}\}}[\lambda_n\cdot \lim\limits_{r\rightarrow 0}\frac{dA_{mn}(r,R_{mk},\theta_{mk})}{dr}] \cdot v
\end{equation}

Based on the expression of $A_{mn}(r,R_{mk},\theta_{mk})$ in equation (\ref{equ_10th}), $\lim\limits_{r\rightarrow 0}\frac{dA_{mn}(r,R_{mk},\theta_{mk})}{dr}$ is derived in Appendix \ref{appendix_02}, and the result is given as equation (\ref{equ_14th}).

\begin{equation}
\label{equ_14th}
\begin{array}{l}
\lim\limits_{r\rightarrow 0}\frac{dA_{mn}(r,R_{mk},\theta_{mk})}{dr} = I(\cos(\theta_{mk})<-\frac{R_{mk}}{R_n^{lb}})[-2\pi \frac{(R_n^{lb})^2}{R_{mk}}\cos(\theta_{mk})]\\
+ I(-\frac{R_{mk}}{R_n^{lb}}<\cos(\theta_{mk})<\frac{R_{mk}}{R_n^{lb}})[-2\frac{(R_n^{lb})^2}{R_{mk}}\cos(\theta_{mk})\arccos(\frac{R_n^{lb}}{R_{mk}}\cos(\theta_{mk}))
+ 2\sqrt{(R_n^{lb})^2-(\frac{(R_n^{lb})^2}{R_{mk}}\cos(\theta_{mk}))^2}]
\end{array}
\end{equation}

\begin{theorem}
\label{theorem_01}
The instantaneous handover rate from \emph{m}th-tier BSs to \emph{n}th-tier BSs for a UE with instantaneous velocity $v$ is $H_k^{m-n}$ and is given by equation (\ref{equ_15th}).

\begin{equation}
\label{equ_15th}
H_k^{m-n}=8\lambda_n\lambda_mv\int\limits_0^{+\infty}\int\limits_0^{\min(1,\frac{R_{mk}}{R_n^{lb}})}
[\sqrt{\frac{1-z^2}{(\frac{R_{mk}}{R_n^{lb}})^2-z^2}}+\sqrt{\frac{(\frac{R_{mk}}{R_n^{lb}})^2-z^2}{1-z^2}}]dz(R_n^{lb})^2
\exp\{-\pi\sum\limits_{i=1}^N(\lambda_i(R_i^{lb})^2)\}dR_{mk}
\end{equation}
\end{theorem}

\begin{proof}
Based on equations (\ref{equ_new05th}) and (\ref{equ_14th}), $H_k^{m-n}$ is further derived in Appendix \ref{appendix_03}.
\end{proof}

Hence, the average handover arrival rate $\lambda_h$ in equation (\ref{equ_02.1th}) can be given by equation (\ref{equ_16th}), where $f_v(v)$ is the probability density function of the velocity $v$ in the specified region with area $S$ and is determined by the walking model of UEs.

\begin{equation}
\label{equ_16th}
\begin{array}{l}
\lambda_h = \sum\limits_{m=1}^N(\mathbb{E}[H_k^mf_uS])=\sum\limits_{m=1}^N\sum\limits_{n=1}^N\mathbb{E}[H_k^{m-n}f_uS]\\
=\!f_u\!S\!\int\limits_{0}^{+\infty}\!vf_v(v)dv\!\sum\limits_{m=1}^N\!\sum\limits_{n=1}^N\!\{\!8\lambda_n\lambda_m\!\!\int\limits_0^{+\infty}\!\int\limits_0^{\min(1,\frac{R_{mk}}{R_n^{lb}})}
\!\![\sqrt{\frac{1\!-\!z^2}{(\frac{R_{mk}}{R_n^{lb}})^2\!-\!z^2}}\!\!+\!\!\sqrt{\frac{(\frac{R_{mk}}{R_n^{lb}})^2\!-\!z^2}{1\!-\!z^2}}]\!dz(R_n^{lb})^2\!
\exp\{\!-\!\pi\sum\limits_{i=1}^N(\lambda_i(R_i^{lb})^2)\}dR_{mk}\}
\end{array}
\end{equation}

Similarly, the average handover arrival rate from \emph{m}th-tier BSs to \emph{n}th-tier BSs  (denoted as $\lambda_h^{m-n}$) can be derived as equation (\ref{equ_17th}).

\begin{equation}
\label{equ_17th}
\begin{array}{l}
\lambda_h^{m-n} = \mathbb{E}[H_k^{m-n}f_uS]\\
=\!f_u\!S\!\int\limits_{0}^{+\infty}\!vf_v(v)dv8\lambda_n\lambda_m\!\!\int\limits_0^{+\infty}\!\int\limits_0^{\min(1,\frac{R_{mk}}{R_n^{lb}})}
\!\![\sqrt{\frac{1\!-\!z^2}{(\frac{R_{mk}}{R_n^{lb}})^2\!-\!z^2}}\!\!+\!\!\sqrt{\frac{(\frac{R_{mk}}{R_n^{lb}})^2\!-\!z^2}{1\!-\!z^2}}]\!dz(R_n^{lb})^2\!
\exp\{\!-\!\pi\sum\limits_{i=1}^N(\lambda_i(R_i^{lb})^2)\}dR_{mk}
\end{array}
\end{equation}

Although the derived handover rate expressions are not closed-form, these expressions can be efficiently computed numerically as opposed to the usual Monte Carlo methods that rely on repeated random sampling to compute these results.

\subsection{Discussions of Handover Rates}

In equations of (\ref{equ_16th}) and (\ref{equ_17th}), the general expressions of handover rates have been derived. In this section, some corollaries and special cases of handover rates in the stochastic modeling of the HCNs would be given.

\begin{corollary}
\label{corollary_01}
$\lambda_h^{m-n}=\lambda_h^{n-m}$ holds for any $m,n\in\{1,\cdots,N\}$, that is, the forward and reverse handover rates between any two tiers are the same.
\end{corollary}
\begin{proof}
See Appendix \ref{appendix_04}.
\end{proof}

Corollary \ref{corollary_01} holds under the condition that the UE movements in different tiers are homogeneous and UEs are uniformly distributed. Corollary \ref{corollary_01} indicates that the mobility between any two tiers would reach statical balance over time.

\begin{corollary}
\label{corollary_02}
When $N=1$, all the BSs are homogeneous and $\{\lambda_n\}=\lambda$, the expression of average handover rate $\lambda_h$ can be further simplified as
\begin{equation}
\label{equ_new07th}
\lambda_h(N=1)=\frac{4\sqrt{\lambda}}{\pi}f_uS\int\limits_0^{+\infty}vf_v(v)dv
\end{equation}
\end{corollary}

\begin{proof}
When $N=1$, according to equation of (\ref{equ_16th}), $\lambda_h$ can be further derived as equation (\ref{equ_new08th}), where (a) is obtained due to $R_m^{lb}=R_{mk}$.

\begin{equation}
\label{equ_new08th}
\begin{array}{l}
\lambda_h(N=1) = \sum\limits_{m=1}^N(\mathbb{E}[H_k^mf_uS])
\mathop=\limits^{(a)}\!f_u\!S\!\int\limits_{0}^{+\infty}\!vf_v(v)dv8\lambda^2\!\!\int\limits_0^{+\infty}\!\int\limits_0^{1}
\!\!2dzR_{mk}^2\!\exp\{\!-\!\pi\lambda R_{mk}^2)\}dR_{mk}\\
=\!f_u\!S\!\int\limits_{0}^{+\infty}\!vf_v(v)dv16\lambda^2\int\limits_0^{+\infty}R_{mk}^2\exp\{\!-\!\pi\lambda R_{mk}^2)\}dR_{mk}\\
=\!f_u\!S\!\int\limits_{0}^{+\infty}\!vf_v(v)dv\frac{8\lambda}{\pi\sqrt{\lambda}}\mathcal{Q}(0)\\
=\frac{4\sqrt{\lambda}}{\pi}f_uS\int\limits_0^{+\infty}vf_v(v)dv
\end{array}
\end{equation}

\end{proof}

Corollary \ref{corollary_02} is consistent with the handover rate expression of homogeneous cellular network given in \cite{cite_last01}\cite{cite_new03}.

\begin{corollary}
\label{corollary_03}
When $\{\alpha_n\}=\alpha$, i.e. BSs in all the tiers have the same path loss exponent, the handover rate of $\lambda_h^{m-n}$ can be simplified as equation (\ref{equ_new09th}),

\begin{equation}
\label{equ_new09th}
\begin{array}{l}
\lambda_h^{m-n}(\{\alpha_n\}=\alpha)\!\!= \!\!\frac{2\lambda_n\lambda_m\beta_n^2f_uS}{\pi(\sum\limits_{i=1}^N\lambda_i\beta_i^2)^{1.5}}\!\int\limits_{0}^{+\infty}\!vf_v(v)dv
\cdot\!\!\!\!\!\!\int\limits_0^{\min(1,\beta_n)}
\!\![\sqrt{\frac{1\!-\!z^2}{\beta_n^2\!-\!z^2}}\!\!+\!\!\sqrt{\frac{\beta_n^2\!-\!z^2}{1\!-\!z^2}}]\!dz
\end{array}
\end{equation}
where $\beta_n\triangleq(\frac{P_nB_n}{P_mB_m})^{\frac{1}{\alpha}}$.
\end{corollary}

\begin{proof}
When $\{\alpha_n\}=\alpha$, according to equation (\ref{equ_04th}), $R_n^{lb}$ can be simplified as $R_n^{lb}=(\frac{P_nB_n}{P_mB_m})^{\frac{1}{\alpha}}R_{mk}=\beta_nR_{mk}$. Then $\lambda_h^{m-n}$ in equation (\ref{equ_17th}) can be further derived as equation (\ref{equ_new10th}).

\begin{equation}
\label{equ_new10th}
\begin{array}{l}
\lambda_h^{m-n}(\{\alpha_n\}=\alpha)
=\!f_u\!S\!\int\limits_{0}^{+\infty}\!vf_v(v)dv8\lambda_n\lambda_m\!\int\limits_0^{+\infty}\!\int\limits_0^{\min(1,\beta_n)}
\!\![\sqrt{\frac{1\!-\!z^2}{\beta_n^2\!-\!z^2}}\!\!+\!\!\sqrt{\frac{\beta_n^2\!-\!z^2}{1\!-\!z^2}}]\!dz\beta_n^2R_{mk}^2\!
\exp\{\!-\!\pi(\sum\limits_{i=1}^N\lambda_i\beta_i^2)R_{mk}^2\}dR_{mk}\\
=\!f_u\!S\!\int\limits_{0}^{+\infty}\!vf_v(v)dv8\lambda_n\lambda_m\beta_n^2\int\limits_0^{\min(1,\beta_n)}
\!\![\sqrt{\frac{1\!-\!z^2}{\beta_n^2\!-\!z^2}}\!\!+\!\!\sqrt{\frac{\beta_n^2\!-\!z^2}{1\!-\!z^2}}]\!dz \!\int\limits_0^{+\infty}R_{mk}^2\!
\exp\{\!-\!\pi(\sum\limits_{i=1}^N\lambda_i\beta_i^2)R_{mk}^2\}dR_{mk}\\
=\!f_u\!S\!\int\limits_{0}^{+\infty}\!vf_v(v)dv8\lambda_n\lambda_m\beta_n^2\int\limits_0^{\min(1,\beta_n)}
\!\![\sqrt{\frac{1\!-\!z^2}{\beta_n^2\!-\!z^2}}\!\!+\!\!\sqrt{\frac{\beta_n^2\!-\!z^2}{1\!-\!z^2}}]\!dz \frac{1}{4\pi(\sum\limits_{i=1}^N\lambda_i\beta_i^2)^{1.5}}\\
=\frac{2\lambda_n\lambda_m\beta_n^2f_uS}{\pi(\sum\limits_{i=1}^N\lambda_i\beta_i^2)^{1.5}}\!\int\limits_{0}^{+\infty}\!vf_v(v)dv\int\limits_0^{\min(1,\beta_n)}
\!\![\sqrt{\frac{1\!-\!z^2}{\beta_n^2\!-\!z^2}}\!\!+\!\!\sqrt{\frac{\beta_n^2\!-\!z^2}{1\!-\!z^2}}]\!dz
\end{array}
\end{equation}

\end{proof}

\begin{corollary}
\label{corollary_04}

For a UE with constant velocity $v$, its residence time in a BS coverage of \emph{m}th-tier (denoted by $T_r^m$) is exponential distributed with average value $\frac{\gamma_m}{H_k^m}$, i.e. the PDF of $T_r^m$ can be given by
\begin{equation}
\label{equ_last01}
f(T_r^m) = \frac{H_k^m}{\gamma_m}\exp(-\frac{H_k^m}{\gamma_m}T_r^m)
\end{equation}
where $\gamma_m$ is the probability that a UE is associated with the \emph{m}th-tier BS and is given by equation (\ref{equ_last04}) referred to \cite{cite_new05}.

\begin{equation}
\label{equ_last04}
\gamma_m = 2\pi\lambda_m \int\limits_0^\infty R_{mk}\exp\{-\pi\sum\limits_{n=1}^N\lambda_n(R_{n}^{lb})^2\}dR_{mk}
\end{equation}

\end{corollary}

\begin{proof}
When the velocity $v$ is constant, the instantaneous handover rate $H_k^m$ would keep constant and do not change with time. So the transition rate from k state to $\overline{k}$ in Figure 2 is constant, i.e. $\frac{1}{\gamma_m}\lim\limits_{t\rightarrow t_0}\frac{dP_{a,k}(t-t_0)}{dt}=-\frac{H_k^m}{\gamma_m}$ is constant for all $t_0\geq 0$. Thus

\begin{equation}
\label{equ_last03}
P_{a,k}(t) = \gamma_m\exp(-\frac{H_k^m}{\gamma_m}t)
\end{equation}
due to $P_{a,k}(0)=\gamma_m$. So,
\begin{equation}
\label{equ_last02}
\begin{array}{l}
 f(T_r^m) = -\frac{1}{\gamma_m}\frac{P_{a,k}(t)}{dt} = \frac{H_k^m}{\gamma_m}\exp(-\frac{H_k^m}{\gamma_m}T_r^m)
\end{array}
\end{equation}

\end{proof}

\begin{corollary}
\label{corollary_05}
When the UEs are not uniformly distributed in the whole region and UEs in each tier BS coverage are uniformly distributed, the handover rate from \emph{m}th-tier BS to \emph{n}th-tier BS in a specified region with area S can be given by equation (\ref{equ_last05}).
\begin{equation}
\label{equ_last05}
\begin{array}{l}
\lambda_h^{m-n}(\{f_{m,u}\}\neq f_u) = \mathbb{E}[H_k^{m-n}f_{m,u}S]\\
=f_{m,u}S\!\int\limits_{0}^{+\infty}\!vf_v(v)dv8\lambda_n\lambda_m\!\!\int\limits_0^{+\infty}\!\int\limits_0^{\min(1,\frac{R_{mk}}{R_n^{lb}})}
\!\![\sqrt{\frac{1\!-\!z^2}{(\frac{R_{mk}}{R_n^{lb}})^2\!-\!z^2}}\!\!+\!\!\sqrt{\frac{(\frac{R_{mk}}{R_n^{lb}})^2\!-\!z^2}{1\!-\!z^2}}]\!dz(R_n^{lb})^2\!
\exp\{\!-\!\pi\sum\limits_{i=1}^N(\lambda_i(R_i^{lb})^2)\}dR_{mk}
\end{array}
\end{equation}
where $f_{m,u}$ is the UE density in the \emph{m}th tier BS coverage. And the total handover rate in the specified region is $\lambda_h(\{f_{m,u}\}\neq f_u)=\sum\limits_{m=1}^N\sum\limits_{n=1}^N\lambda_h^{m-n}(\{f_{m,u}\}\neq f_u)$.
\end{corollary}

\begin{proof}
For a UE admitted to a \emph{m}th-tier BS, its average handover rate to a \emph{n}th-tier BS is $\frac{H_k^m}{\gamma_m}$. On the other hand, the average coverage area of \emph{m}th-tier BSs in the specified region is $S\gamma_m$. So the total handover rate from \emph{m}th-tier BSs to \emph{n}th-tier BSs in the specified region is $\lambda_h^{m-n} = \mathbb{E}[\frac{H_k^{m-n}}{\gamma_m}f_{m,u}S\gamma_m]=\mathbb{E}[H_k^{m-n}f_{m,u}S]$.
\end{proof}

\begin{corollary}
\label{corollary_06}
When the walking models of UEs in different tier BS coverage are different, the handover rate from the \emph{m}th-tier BS to \emph{n}th-tier BS in a specified region with area S can be given by equation (\ref{equ_last06}).
\begin{equation}
\label{equ_last06}
\begin{array}{l}
\lambda_h^{m-n}(\{f_{m,v}(v)\}\neq f_v(v)) = \mathbb{E}[H_k^{m-n}f_uS]\\
=f_uS\!\int\limits_{0}^{+\infty}\!vf_{m,v}(v)dv8\lambda_n\lambda_m\!\!\int\limits_0^{+\infty}\!\int\limits_0^{\min(1,\frac{R_{mk}}{R_n^{lb}})}
\!\![\sqrt{\frac{1\!-\!z^2}{(\frac{R_{mk}}{R_n^{lb}})^2\!-\!z^2}}\!\!+\!\!\sqrt{\frac{(\frac{R_{mk}}{R_n^{lb}})^2\!-\!z^2}{1\!-\!z^2}}]\!dz(R_n^{lb})^2\!
\exp\{\!-\!\pi\sum\limits_{i=1}^N(\lambda_i(R_i^{lb})^2)\}dR_{mk}
\end{array}
\end{equation}
where $f_{m,v}(v)$ is the velocity distribution in the \emph{m}th-tier BS coverage. And the total handover rate in the specified region is $\lambda_h(\{f_{m,v}(v)\}\neq f_v(v))=\sum\limits_{m=1}^N\sum\limits_{n=1}^N\lambda_h^{m-n}(\{f_{m,v}(v)\}\neq f_v(v))$.

\end{corollary}

\begin{proof}
According to the derivation of $H_k^{m-n}$, $H_k^{m-n}$ is only related to the velocity distribution in the \emph{m}th-tier BS coverage and is independent of the velocity distribution in other tier BS coverage. Thus, Corollary \ref{corollary_06} is can be directly derived from the derivation of $H_k^{m-n}$ in the former section.
\end{proof}

From the above analysis, we could give the general steps for handover rate analysis by stochastic geometry modeling: firstly, obtain the PDF of the associated BS's position, secondly, by the infinitesimal method, get the area of the bad region, based on which, then derive the instantaneous handover probability, and derive the instantaneous handover rate through taking the derivative, and at last, average the instantaneous handover rate by the distribution of UE's velocity.

\section{Simulation Results}

\subsection{Validation of the Analysis}

Now that we have developed the general expression of handover rate for HCNs by stochastic geometry modeling, it is important to see how well the analytical results match the computer simulation. Here, we consider two types of walking models, i.e. straight-line walking model and RWP walking model. For the straight-line walking model, UE would move without changing its moving direction, while for the RWP model, UE would change its moving direction for a randomly chosen direction of $[0,2\pi]$ and then keep the direction for a randomly chosen duration of [0,100s]. For both models, the velocity is uniformly distributed in $[0,2v]$.

In the simulation, we considered a circle region with radius of 10Km, generated BSs of each tier and their positions by the PPP $\Phi_n$ with the density of $\lambda_n$ for each try, and generated the UEs with the density of $f_u$. For each try, all the UEs move with a specified walking model for 10$^4$ seconds. We counted the number of handover in a specified region with area S = 1km$^2$ and gave the handover rates by averaging the results.

A total of two tiers are modeled according to PPP in the simulation (N=2). For the simulation, some parameters keep constant and these parameters are $f_u$ = 100/km$^2$, S = 1km$^2$, $P_1=1$, $P_2 = 0.2$, $B_1=1$, $\alpha_1 = 3.5$ and $\lambda_1=1$/km$^2$. Other parameters would vary in different tries and these parameters are the velocity $v$, the 2nd-tier BS density $\lambda_2$ and, 2nd-tier BS's path loss factor $\alpha_2$ and the 2nd-tier BS's bias factor $B_2$.

Figure 4 compares the analytical and experimental handover rates of $\lambda_h^{1-1}$, $\lambda_h^{1-2}$ and $\lambda_h^{2-2}$. In Figure 4, the dotted lines represent the analytical results, the circles and triangles represent the simulation results obtained by straight-line walking model and RWP walking model, respectively, And we set $\lambda_2=2$/km$^2$, $B_2=1$ and $\alpha_2=3.5$. In Figure 4, it is observed that the handover rates of $\lambda_h^{1-1}$, $\lambda_h^{1-2}$ and $\lambda_h^{2-2}$ increase linearly with the average velocity of UE. This matches the derived expression of handover rate.  In Figure 4, we can see that $\lambda_h^{1-1}$ is the largest, $\lambda_h^{1-2}$ takes the second place and $\lambda_h^{2-2}$ is the minimum for any value of average velocity. This can be explained as follows. Since $P_1=1$ and $P_2=0.2$, the coverage of 1st-tier BS is much larger than the coverage of 2nd-tier BS. So the 1st-tier BS has longer boundary line than the 2nd-tier BS. Thus, the total length of the boundary line between two 1st-tier BSs is the longest, the total length of the boundary line between a 1st-tier BS and a 2nd-tier BS is shorter and the total length of boundary line between two 2nd-tier BSs is the
shortest. Hence, the handover is most likely to happen at the boundary between two 1st-tier BSs, is less likely to happen at the boundary between a 1st-tier BS and a 2nd-tier BS, and is least likely to happen at the boundary
between two 2nd-tier BSs.  From the figure, it can be seen that the relative error between the analytical results and the simulation results are less than 3\% for both walking models. The relative error is most likely brought by the limited simulation time. The good matching validates that the analysis is reliable with the variation of UE velocity. In the simulation, we find that the straight-line walking model and the RWP walking model without pause time almost have the same handover rates when they have the same velocity distribution.

Figure 5 illustrates the analytical and experimental results of the total handover rate ($\lambda_h$) and the handover rate between the two tiers ($\lambda_h^{1-2}$) with different 2nd-tier BSs density ($\lambda_2$) for different average UE velocity ($v$).
From the figure, we can observe that handover rates increase linearly with the average UE velocity, and both
$\lambda_h$ and $\lambda_h^{1-2}$ increase with the 2nd-tier BS density. This is because larger BS density leads to smaller cells and UEs are more likely to move out the smaller cells.
The analytical results match well with experimental results of straight-line walking model and RWP walking model, with relative error less than 3\%. This validates that the analysis is reliable with the variation of both BS density and moving velocity.

Figure 6 demonstrates the analytical and experimental results of handover rates with variations of 2nd-tier BS density ($\lambda_2$) and path loss factor ($\alpha_2$). The analytical results match well with the experimental results for both walking models with relative error less than 3\%. This validates that the analysis is reliable with the variations of both BS density and path loss factors. From this figure, we can observe that all the handover rates increase with the 2nd-tier BS density. The handover rate of $\lambda_h^{1-1}$ increases with 2nd-tier BS path loss factor $\alpha_2$, $\lambda_h^{1-2}$ increases with $\alpha_2$ firstly, then decreases with $\alpha_2$, and $\lambda_h^{2-2}$ decreases with $\alpha_2$. This is because the coverage of the 2nd-tier BS decreases with $\alpha_2$ and the coverage of the 1st-tier BS increases with $\alpha_2$ according to equation (\ref{equ_new01th}). Thus, the boundary line of the 1st-tier BSs increases with the $\alpha_2$ and the boundary line of the 2nd-tier BSs decreases with the $\alpha_2$. So, the total length of the boundary line between two 1st-tier BSs increases with $\alpha_2$, the total length of the boundary line between  two 2nd-tier BSs decreases with $\alpha_2$, and the total length of the boundary line between a 1st-tier BS and a 2nd-tier BS increases firstly with $\alpha_2$, then decreases with $\alpha_2$.
It is the reason that leads to the increasing of $\lambda_h^{1-1}$, the decreasing of $\lambda_h^{2-2}$, and the changing of $\lambda_h^{1-2}$.

Figure 7 demonstrates the experimental results of the forward and reverse handover rates between the two tiers for the RWP walking model with the variations of 2nd-tier BS density, path loss factor and average velocity. Figure 7 shows that the forward handover rates ($\lambda_h^{1-2}$) match well with reverse handover rates ($\lambda_h^{2-1}$) between the two tiers with relative error less than 3\%. This validates the reliability of Corollary \ref{corollary_01}.

Figure 8 demonstrates the analytical and experimental results of residual time distributions with constant velocity. The analysis results (denoted by `Theo.' in the figure) is calculated by equation (\ref{equ_last01}), and the experimental results (denoted by `Simu.' in the figure) is obtained by the RWP model with constant velocity $v$=5m/s. From the figure, we can see that the experimental results match well with the analysis results, which validates the effectiveness of Corollary \ref{corollary_04}. As the figure illustrated, the residual time in the 1st-tier BS is much larger than the residual time in the 2nd-tier BS due to the fact that the 1st-tier BS coverage is much larger than the 2nd-tier BS coverage. As the figure shown, both the residual time in the 1st-tier BS and in the 2nd-tier BS decrease with the 2nd-tier BS density $\lambda_2$, but the decrement of residual time in the 1st-tier BS is more significant. This can be explained as follows. With the increase of 2nd-tier BS density, the average coverage of each tier BS decreases, so the residual time decreases for each tier. On the other hand, with the increase of 2nd-tier BS density, the average 1st-tier BS coverage decreases more significantly because the coverage of each 1st-tier BS is much larger than the 2nd-tier BS coverage and the increased 2nd-tier BSs would be more likely to occupy the 1st-tier BSs coverage. Thus, the decrement of residual time in the 1st-tier BS is more significant.

\subsection{Effect of Bias Factor}

Figure 9 shows the numerical results of handover rates with the variation of 2nd-tier BS's bias factor ($B_2$). From this figure, we can observe that the total handover rate ($\lambda_h$) decreases with the 2nd-tier BS's bias factor, $\lambda_h^{1-1}$ and $\lambda_h^{1-2}$ decrease with $B_2$ and $\lambda_h^{2-2}$ increases with $B_2$ when $B_2\in[1,2]$, which is a reasonable range of $B_2$. This can be explained as follows.
According to equation (\ref{equ_16th}),
when $R_{mk}$ equals $R_{n}^{lb}$, $\lambda_h^{m-n}$ could reach the minimal value, and $\lambda_h^{m-n}$ decreases with $R_{mk}$ if $R_{mk}<R_n^{lb}$ holds.
Thus, when $B_2$ increases and $R_{1k}<R_2^{lb}$ holds, $R_{1k}$ increases and $\lambda_h^{1-2}$ decreases. When $B_2$ increases, the coverage of 1st-tier BS decreases and the coverage of 2nd-tier BS increases according to equation (\ref{equ_new01th}), so the boundary line of 1st-tier BS decreases and the boundary line of the 2nd-tier BS increases, which lead to the results that $\lambda_h^{1-1}$ decreases with $B_2$ and $\lambda_h^{2-2}$ increases with $B_2$. The results demonstrate that we can decrease the total handover rate by reasonably adjusting the bias factors.

\section{Conclusion}
In the literature, there has not been any general handover rate derivation for the heterogeneous cellular networks. Thus, in this paper, we give a generalized handover analytical framework by employing the stochastic geometry modeling for heterogeneous cellular networks, derive the arithmetic expression of handover rate, give some meaningful corollories and validate the analysis by computer simulation. The analysis may shed some light for future extension and study.

\section*{Acknowledgements}
  \ifthenelse{\boolean{publ}}{\small}{}
    The authors wish to thank the anonymous reviewers for their constructive comments. This paper is supported by National programs for High Technology Research and Development (2012AA01A50603).

\section{Appendix}

\subsection{Boundary conditions of bad region}\label{appendix_01}

According to the bad region boundary conditions in equation (\ref{equ_08th}), we can further derive $\vartheta_{nj}$ as follow:
\begin{equation}
\label{equ_app01}
\left\{\begin{array}{lc}
\vartheta_{nj}=0, &\frac{R_{nj}^2+x_{nj}}{2rR_{nj}}> 1 \\
\vartheta_{nj}=2\arccos\frac{R_{nj}^2+x_{nj}}{2rR_{nj}}, &-1 \leq \frac{R_{nj}^2+x_{nj}}{2rR_{nj}} \leq1\\
\vartheta_{nj}=2\pi, &\frac{R_{nj}^2+x_{nj}}{2rR_{nj}}<-1
\end{array}
\right.
\end{equation}

For the case of $-1 \leq \frac{R_{nj}^2+x_{nj}}{2rR_{nj}} \leq 1$, we can further derive the inequations of $\frac{R_{nj}^2+x_{nj}}{2rR_{nj}} \leq 1$ and $\frac{R_{nj}^2+x_{nj}}{2rR_{nj}} \geq -1$  as equations (\ref{equ_app02.1}) and (\ref{equ_app02.2}), respectively,
\begin{equation}
\label{equ_app02.1}
\begin{array}{l}
R_{nj}^2-2rR_{nj}+x_{nj}\leq 0\\
\Rightarrow
R_{nj}\!\in\! [r-\sqrt{r^2-x_{nj}},r+\sqrt{r^2-x_{nj}}]
\end{array}
\end{equation}

\begin{equation}
\label{equ_app02.2}
\begin{array}{l}
R_{nj}^2+2rR_{nj}+x_{nj}\geq 0\\
\Rightarrow
R_{nj}\!\in\! (\!-\!\infty, \!-\!r\!-\!\sqrt{r^2\!-\!x_{nj}}) \!\bigcup(-\!r \!+\! \sqrt{r^2\!-\!x_{nj}},\!+\!\infty)
\end{array}
\end{equation}

Hence, we can give the range of $R_{nj}$ as equation (\ref{equ_app02.3}) according to equation (\ref{equ_04th}) and $R_{nj}>0$.

\begin{equation}
\label{equ_app02.3}
\min(R_n^{lb},|\!-\!r\!+\!\sqrt{r^2\!-\!x_{nj}}|)\!\leq\! R_{nj} \!\leq\! r\!+\!\sqrt{r^2\!-\!x_{nj}}
\end{equation}

As $-r+\sqrt{r^2\!-\!x_{nj}}>0$ when $r\rightarrow 0^+$, we can further simplify equation (\ref{equ_app02.3}) as
\begin{equation}
\label{equ_app03}
\min(R_n^{lb},\!-\!r\!+\!\sqrt{r^2\!-\!x_{nj}})\!\leq\! R_{nj} \!\leq\! r\!+\!\sqrt{r^2\!-\!x_{nj}}
\end{equation}

For case of $\frac{R_{nj}^2+x_{nj}}{2rR_{nj}}> 1$, similar to equation (\ref{equ_app02.1}), we can further derive the inequation of $\frac{R_{nj}^2+x_{nj}}{2rR_{nj}}> 1$ as $R_{nj}\in (\max(R_{nj}^{lb},r+\sqrt{r^2-x_{nj}}),+\infty)$.

For the case of $\frac{R_{nj}^2+x_{nj}}{2rR_{nj}}<-1$, similar to equation (\ref{equ_app02.2}), we can further derive the inequation of $\frac{R_{nj}^2+x_{nj}}{2rR_{nj}}<-1$ as $R_{nj}\in (R_{nj}^{lb},-r+\sqrt{r^2-x_{nj}})$

On the other hand, according to the expression of $R_n^{lb}$ in equation (\ref{equ_04th}) and $x_{nj}$ in equation (\ref{equ_new04th}), $R_n^{lb}=-r+\sqrt{r^2\!-\!x_{nj}}$ holds when $r=0$. And the derivative of $-r+\sqrt{r^2\!-\!x_{nj}}$ at $r=0$ is derived as
\begin{equation}
\label{equ_app04}
\begin{array}{l}
\lim\limits_{r\rightarrow 0}\frac{d(-r+\sqrt{r^2\!-\!x_{nj}})}{dr}
=(\frac{B_nP_n}{B_mP_m})^{\frac{1}{\alpha_n}}(R_{mk})^{\frac{\alpha_m}{\alpha_n}}\frac{-\cos(\theta_{mk})}{R_{mk}}-1
=-\frac{R_n^{lb}}{R_{mk}}\cos(\theta_{mk})-1
\end{array}
\end{equation}

Thus, $-r+\sqrt{r^2\!-\!x_{nj}}$ increases at $r=0$ if $\cos(\theta_{mk})<-\frac{R_{mk}}{R_n^{lb}}$ or decreases otherwise. So,
\begin{equation}
\label{equ_app05}
\left\{\begin{array}{lc}
\!\!R_n^{lb} \!>\! -r\!+\!\sqrt{r^2\!-\!x_{nj}}, & \!\!\cos(\theta_{mk})\!>\!-\frac{R_{mk}}{R_n^{lb}}\\
\!\!R_n^{lb} \!<\! -r\!+\!\sqrt{r^2\!-\!x_{nj}}, & \!\!\cos(\theta_{mk})\!<\!-\frac{R_{mk}}{R_n^{lb}}
\end{array}
\right.
\end{equation}

Similarly, the following relationships can be derived,
\begin{equation}
\label{equ_app06}
\left\{\begin{array}{lc}
\!\!R_n^{lb} \!>\! r\!+\!\sqrt{r^2\!-\!x_{nj}}, & \!\!\cos(\theta_{mk})\!>\!\frac{R_{mk}}{R_n^{lb}}\\
\!\!R_n^{lb} \!<\! r\!+\!\sqrt{r^2\!-\!x_{nj}}, & \!\!\frac{-R_{mk}}{R_n^{lb}}\!\!<\cos(\theta_{mk})\!<\!\frac{R_{mk}}{R_n^{lb}}
\end{array}
\right.
\end{equation}

Based on the above relationships, the boundary conditions of equation (\ref{equ_app01}) can be derived as equation (\ref{equ_09th}).

\subsection{Derivation of $\lim\limits_{r\rightarrow 0}\frac{dA_{mn}(r,R_{mk},\theta_{mk})}{dr}$}\label{appendix_02}

For the clarity of expression, we define the expressions of $A_{mn}^{(1)}(r,R_{mk},\theta_{mk})$, $A_{mn}^{(2)}(r,R_{mk},\theta_{mk})$ and $A_{mn}^{(3)}(r,R_{mk},\theta_{mk})$ in the equation of (\ref{equ_app07}).
\begin{equation}
\label{equ_app07}
\begin{array}{c}
A_{mn}^{(1)}(r,R_{mk},\theta_{mk})=[\int\limits_{R_n^{lb}}^{r+\sqrt{r^2-x_{nj}}}2\arccos(\frac{R_{nj}^2+x_{nj}}{2rR_{nj}})R_{nj}dR_{nj}]I(-\frac{R_{mk}}{R_n^{lb}}\leq\cos(\theta_{mk})\leq\frac{R_{mk}}{R_n^{lb}})\\
A_{mn}^{(2)}(r,R_{mk},\theta_{mk})=[\int\limits_{-r+\sqrt{r^2-x_{nj}}}^{r+\sqrt{r^2-x_{nj}}}2\arccos(\frac{R_{nj}^2+x_{nj}}{2rR_{nj}})R_{nj}dR_{nj}]I(\cos(\theta_{mk})<-\frac{R_{mk}}{R_n^{lb}})\\
A_{mn}^{(3)}(r,R_{mk},\theta_{mk})=[ \int\limits_{R_n^{lb}}^{-r+\sqrt{r^2-x_{nj}}}2\pi R_{nj}dR_{nj} ]I(\cos(\theta_{mk})<-\frac{R_{mk}}{R_n^{lb}})
\end{array}
\end{equation}

Based on those expressions and equation (\ref{equ_10th}), the following relationship can be obtained,
\begin{equation}
\label{equ_app08}
A_{mn}=A_{mn}^{(1)}+A_{mn}^{(2)}+A_{mn}^{(3)}
\end{equation}

Then $\lim\limits_{r\rightarrow 0}\frac{dA_{mn}(r,R_{mk},\theta_{mk})}{dr}$ can be derived by deriving $\lim\limits_{r\rightarrow 0}\frac{dA_{mn}^{(i)}(r,R_{mk},\theta_{mk})}{dr}$, i = 1,2,3. For simplicity, we define $\varphi^u(r)$, $\varphi^d(r)$ and $f(r,R_{nj})$ as $\varphi^u(r)=r+\sqrt{r^2-x_{nj}}$, $\varphi^d(r)=-r+\sqrt{r^2-x_{nj}}$, and $f(r,R_{nj})=2R_{nj}\arccos(\frac{x_{nj}+R_{nj}^2}{2rR_{nj}})$, respectively.

\begin{equation}
\label{equ_app09}
\begin{array}{l}
\lim\limits_{r\rightarrow 0}\frac{dA_{mn}^{(1)}}{dr}=\lim\limits_{r\rightarrow 0}\frac{d}{dr}\int\limits_{R_n^{lb}}^{\varphi^u(r)}f(r,R_{nj})dR_{nj}I^{(1)}\\
=\!\!\lim\limits_{r\rightarrow 0}[f(r,\varphi^u(r))\frac{d}{dr}\varphi^u(r)\!\!+\!\!\!\!\int\limits_{R_n^{lb}}^{\varphi^u(r)}\!\frac{d}{dr}f(r,R_{nj})dR_{nj}]I^{(1)}\\
\mathop = \limits^{(a)}\lim\limits_{r\rightarrow 0}[\int\limits_{R_n^{lb}}^{\varphi^u(r)}\!\frac{d}{dr}f(r,R_{nj})dR_{nj}]I^{(1)}\\
=\lim\limits_{r\rightarrow 0}[\int\limits_{R_n^{lb}}^{\varphi^u(r)}\frac{-2R_{nj}}{\sqrt{1-(\frac{R_{nj}^2+x_{nj}}{2R_{nj}r})^2}}\frac{\frac{dx_{nj}}{dr}r-(R_{nj}^2+x_{nj})}{2R_{nj}r^2}dR_{nj}]I^{(1)}\\
\mathop=\limits^{y=R_{nj}^2}\lim\limits_{r\rightarrow 0}I^{(1)}[\int\limits_{(R_n^{lb})^2}^{(\varphi^u(r))^2}\frac{-1}{\sqrt{1-\frac{(y+x_{nj})^2}{4yr^2}}}\frac{\frac{dx_{nj}}{dr}r-(y+x_{nj})}{r^2}\frac{1}{2\sqrt{y}}dy]\\
=\lim\limits_{r\rightarrow 0}I^{(1)}[\int\limits_{(R_n^{lb})^2}^{(\varphi^u(r))^2}\frac{-1}{\sqrt{4r^4-4r^2x_{nj}-(y+x_{nj}-2r^2)^2}}
\cdot \frac{(\frac{dx_{nj}}{dr}r-2r^2)-(y+x_{nj}-2r^2)}{r}dy]\\
=\lim\limits_{r\rightarrow 0}-I^{(1)}[-\frac{\frac{dx_{nj}}{dr}r-2r^2}{r}\arccos\frac{y+x_{nj}-2r^2}{\sqrt{4r^4-4r^2x_{nj}}}
+\frac{1}{r}\sqrt{4r^4-4r^2x_{nj}-(y+x_{nj}-2r^2)^2}]|_{(R_n^{lb})^2}^{(\varphi^u(r))^2}\\
=\lim\limits_{r\rightarrow 0}-I^{(1)}[\frac{\frac{dx_{nj}}{dr}r-2r^2}{r}\arccos(\frac{(R_n^{lb})^2+x_{nj}-2r^2}{\sqrt{4r^4-4r^2x_{nj}}})
-\frac{1}{r}\sqrt{4r^4-4r^2x_{nj}-((R_n^{lb})^2+x_{nj}-2r^2)^2}]\\
\mathop = \limits^{(b)}I^{(1)}[-\frac{2(R_n^{lb})^2}{R_{mk}}\cos(\theta_{mk})\arccos(\frac{R_n^{lb}}{R_{mk}}\cos(\theta_{mk}))
+2\sqrt{(R_n^{lb})^2-(\frac{(R_n^{lb})^2}{R_{mk}}\cos(\theta_{mk}))^2}]
\end{array}
\end{equation}
where $I^{(1)}=I(-\frac{R_{mk}}{R_n^{lb}}\leq\cos(\theta_{mk})\leq\frac{R_{mk}}{R_n^{lb}})$, and $(a)$ is given by $f(r,\varphi^u(r))=0$, $(b)$ is given due to $\lim\limits_{r\rightarrow 0}x_{nj}=-(R_n^{lb})^2$, $\lim\limits_{r\rightarrow 0}\frac{d}{dr}x_{nj}=\frac{2(R_n^{lb})^2}{R_{mk}}\cos(\theta_{mk})$ and the L'Hopital rule.

\begin{equation}
\label{equ_app10}
\begin{array}{l}
\lim\limits_{r\rightarrow 0}\frac{dA_{mn}^{(2)}}{dr}=\lim\limits_{r\rightarrow 0}\int\limits_{\varphi^d(r)}^{\varphi^u(r)}f(r,R_{nj})dR_{nj}I^{(2)}\\
=\lim\limits_{r\rightarrow 0}[f(r,\varphi^u(r))\frac{d}{dr}\varphi^u(r)-f(r,\varphi^d(r))\frac{d}{dr}\varphi^d(r)
+\int\limits_{\varphi^d(r)}^{\varphi^u(r)}\frac{d}{dr}f(r,R_{nj})dR_{nj}]I^{(2)}\\
\mathop = \limits^{(c)}[2\pi R_n^{lb}(\frac{R_n^{lb}}{R_{mk}}\cos(\theta_{mk})\!+1\!)\!-\!2\pi\frac{(R_n^{lb})^2}{R_{mk}}\cos(\theta_{mk})]I^{(2)}
\end{array}
\end{equation}
where $I^{(2)}=I(\cos(\theta_{mk})<-\frac{R_{mk}}{R_n^{lb}})$, and $(c)$ is obtained by $f(r,\varphi^u(r))=0$, $f(r,\varphi^d(r))=2\pi$, $\lim\limits_{r\rightarrow 0}\frac{d}{dr}\varphi^d(r)=-\frac{R_n^{lb}}{R_{mk}}\cos(\theta_{mk})-1$ and $\lim\limits_{r\rightarrow 0}\int\limits_{\varphi^d(r)}^{\varphi^u(r)}\frac{d}{dr}f(r,R_{nj})dR_{nj}=-\!2\pi\frac{(R_n^{lb})^2}{R_{mk}}\cos(\theta_{mk})$, which can be derived similarly as equation (\ref{equ_app09}).

\begin{equation}
\label{equ_app11}
\begin{array}{l}
\lim\limits_{r\rightarrow 0}\frac{dA_{mn}^{(3)}}{dr}=\lim\limits_{r\rightarrow 0}\frac{d}{dr}\int\limits_{R_n^{lb}}^{\varphi^d(r)}2\pi R_{nj}dR_{nj}I^{(2)}\\
=\lim\limits_{r\rightarrow 0}[2\pi \varphi^d(r)\frac{d\varphi^d(r)}{dr}]I^{(2)}\\
=-2\pi R_n^{lb}(\frac{R_n^{lb}}{R_{mk}}\cos(\theta_{mk})+1)I^{(2)}
\end{array}
\end{equation}
which is obtained due to $\lim\limits_{r\rightarrow 0}\varphi^d(r) = R_n^{lb}$.

Thus, add equations of (\ref{equ_app09}) (\ref{equ_app10}) and (\ref{equ_app11}) together, we can give $\lim\limits_{r\rightarrow 0}\frac{dA_{mn}(r,R_{mk},\theta_{mk})}{dr}$ as equation (\ref{equ_14th}).

\subsection{Derivation of $H_k^{m-n}$}\label{appendix_03}

According equation (\ref{equ_new05th}), $H_k^{m-n}$ can be further expressed as equation (\ref{equ_app13}),

\begin{equation}
\label{equ_app13}
\begin{array}{l}
H_k^{m-n}\!\!=\!\!\lambda_nv\!\!\int\limits_0^{+\infty}\!\int\limits_0^{2\pi}\lim\limits_{r\rightarrow 0}\!\frac{dA_{mn}}{dr}\!f(\theta_{mk})\!f(R_{mk})d\theta_{mk}dR_{mk}\\
=\!\!\lambda_nv\!\!\int\limits_0^{+\infty}\!\int\limits_0^{2\pi}\lim\limits_{r\rightarrow 0}\!\frac{dA_{mn}}{dr}\!\lambda_m\!R_{mk}\!\exp(\!-\!\pi\!\lambda_m\!R_{mk}^2\!)d\theta_{mk}dR_{mk}
\end{array}
\end{equation}

For simplicity, we define $h_{k,1}^{m-n}$, $h_{k,2}^{m-n}$ and $h_{k,3}^{m-n}$ in equation of (\ref{equ_app12}).

\begin{equation}
\label{equ_app12}
\begin{array}{l}
h_{k,1}^{m-n}=I^{(2)}[-2\pi\frac{(R_n^{lb})^2}{R_{mk}}\cos(\theta_{mk})]\\
h_{k,2}^{m-n}\!\!=\!\!I^{(1)}[\!-\!2\frac{(R_n^{lb})^2}{R_{mk}}\cos(\theta_{mk})\!\arccos(\frac{R_n^{lb}}{R_{mk}}\cos(\theta_{mk}))]\\
h_{k,3}^{m-n}=I^{(1)}[2\sqrt{(R_n^{lb})^2-(\frac{(R_n^{lb})^2}{R_{mk}}\cos(\theta_{mk}))^2}]
\end{array}
\end{equation}

Thus
\begin{equation}
\label{equ_app14}
H_k^{m-n} = \mathbb{E}[h_{k,1}^{m-n}+h_{k,2}^{m-n}+h_{k,3}^{m-n}]
\end{equation}

And $\mathbb{E}[h_{k,1}^{m-n}]$, $\mathbb{E}[h_{k,2}^{m-n}]$ and $\mathbb{E}[h_{k,3}^{m-n}]$ can be derived as equations of (\ref{equ_app15}), (\ref{equ_app16}) and (\ref{equ_app17}), respectively.

\begin{equation}
\label{equ_app15}
\begin{array}{l}
\mathbb{E}[h_{k,1}^{m-n}]=2\lambda_n\lambda_mv\int\limits_0^{+\infty}\int\limits_{\pi-\arccos(\min(1,\frac{R_{mk}}{R_n^{lb}}))}^{\pi}-2\pi\frac{(R_n^{lb})^2}{R_{mk}}\cos(\theta_{mk})R_{mk}\exp\{\!-\!\pi\sum\limits_{i=1}^N(\lambda_i(R_i^{lb})^2)\}d\theta_{mk}dR_{mk}\\
=2\lambda_n\lambda_mv\int\limits_0^{+\infty}[2\pi\frac{(R_n^{lb})^2}{R_{mk}}(-\sin(\theta_{mk}))|_{\pi-\arccos(\min(1,\frac{R_{mk}}{R_n^{lb}}))}^{\pi}]R_{mk}\exp\{\!-\!\pi\sum\limits_{i=1}^N(\lambda_i(R_i^{lb})^2)\}dR_{mk}\\
=2\lambda_n\lambda_mv\int\limits_0^{+\infty}[2\pi\frac{(R_n^{lb})^2}{R_{mk}}\sqrt{1-(\min(1,\frac{R_{mk}}{R_n^{lb}}))^2}]R_{mk}\exp\{\!-\!\pi\sum\limits_{i=1}^N(\lambda_i(R_i^{lb})^2)\}dR_{mk}
\end{array}
\end{equation}

\begin{equation}
\label{equ_app16}
\small
\begin{array}{l}
\mathbb{E}[h_{k,2}^{m-n}]\mathop = \limits^{z=\cos(\theta_{mk})} 2\lambda_n\lambda_mv\int\limits_0^{+\infty}\int\limits_{\min(1,\frac{R_{mk}}{R_n^{lb}})}^{-\min(1,\frac{R_{mk}}{R_n^{lb}})}[-2\frac{(R_n^{lb})^2}{R_{mk}}z\arccos(\frac{R_n^{lb}}{R_{mk}}z)]\frac{-1}{\sqrt{1-z^2}}dzR_{mk}\exp\{\!-\!\pi\sum\limits_{i=1}^N(\lambda_i(R_i^{lb})^2)\}dR_{mk}\\
=2\lambda_n\lambda_mv\int\limits_0^{+\infty}2\frac{(R_n^{lb})^2}{R_{mk}}[(-\sqrt{1-z^2}\arccos(\frac{R_n^{lb}}{R_{mk}}z))|_{\min(1,\frac{R_{mk}}{R_n^{lb}})}^{-\min(1,\frac{R_{mk}}{R_n^{lb}})}+2\int\limits_0^{\min(1,\frac{R_{mk}}{R_n^{lb}})}\frac{\sqrt{1-z^2}}{\sqrt{(\frac{R_{mk}}{R_n^{lb}})^2-z^2}}]R_{mk}\exp\{\!-\!\pi\sum\limits_{i=1}^N(\lambda_i(R_i^{lb})^2)\}dR_{mk}\\
=2\lambda_n\lambda_mv\int\limits_0^{+\infty}2\frac{(R_n^{lb})^2}{R_{mk}}[-\pi\sqrt{1-(\min(1,\frac{R_{mk}}{R_n^{lb}}))^2}+2\int\limits_0^{\min(1,\frac{R_{mk}}{R_n^{lb}})}\frac{\sqrt{1-z^2}}{\sqrt{(\frac{R_{mk}}{R_n^{lb}})^2-z^2}}]R_{mk}\exp\{\!-\!\pi\sum\limits_{i=1}^N(\lambda_i(R_i^{lb})^2)\}dR_{mk}
\end{array}
\end{equation}

\begin{equation}
\label{equ_app17}
\begin{array}{l}
\mathbb{E}[h_{k,3}^{m-n}]\mathop = \limits^{z=\cos(\theta_{mk})} 2\lambda_n\lambda_mv\int\limits_0^{+\infty}\int\limits_{\min(1,\frac{R_{mk}}{R_n^{lb}})}^{-\min(1,\frac{R_{mk}}{R_n^{lb}})}
[2\sqrt{1-(\frac{R_n^{lb}}{R_{mk}}z)^2}\frac{-1}{\sqrt{1-z^2}}R_n^{lb}]R_{mk}\exp\{\!-\!\pi\sum\limits_{i=1}^N(\lambda_i(R_i^{lb})^2)\}dR_{mk}\\
=2\lambda_n\lambda_mv\int\limits_0^{+\infty}2\frac{(R_n^{lb})^2}{R_{mk}}
[2\int\limits_0^{\min(1,\frac{R_{mk}}{R_n^{lb}})}\frac{\sqrt{(\frac{R_{mk}}{R_n^{lb}})^2-z^2}}{\sqrt{1-z^2}}]R_{mk}\exp\{\!-\!\pi\sum\limits_{i=1}^N(\lambda_i(R_i^{lb})^2)\}dR_{mk}
\end{array}
\end{equation}

Thus, $H_k^{m-n}$ can be derived as equation of (\ref{equ_15th}) by adding the results of equations of (\ref{equ_app15}), (\ref{equ_app16}) and (\ref{equ_app17}) together.

\subsection{Proof of Corollary \ref{corollary_01}}\label{appendix_04}

According to the equation (\ref{equ_17th}), $\lambda_h^{m-n}=\lambda_h^{n-m}$ can be proved by proving that the equation of $H_k^{m-n}=H_k^{n-m}$ holds. Based on the expression of $H_k^{m-n}$ in equation (\ref{equ_15th}), the proof can be given by the equation (\ref{equ_app18}), where (d) follows from plugging $R_{mk}=x^{\frac{\alpha_n}{\alpha_m}}(\frac{P_mB_m}{P_nB_n})^{\frac{1}{\alpha_m}}$, (e) follows from plugging $z_1=\frac{R_{m,n}^{lb}}{x}z$, and $R_{i,n}^{lb}=(\frac{P_iB_i}{P_nB_n})^{\frac{1}{\alpha_i}}x^{\frac{\alpha_n}{\alpha_i}}$, (f) follows from plugging $x=R_{nj}$, where $R_{nj}$ is the nearest distance of the \emph{n}th-tier BSs to the origin.

\begin{equation}
\label{equ_app18}
\begin{array}{l}
H_k^{m-n}=8\lambda_n\lambda_mv\int\limits_0^{+\infty}\int\limits_0^{\min(1,\frac{R_{mk}}{R_n^{lb}})}
[\sqrt{\frac{1-z^2}{(\frac{R_{mk}}{R_n^{lb}})^2-z^2}}+\sqrt{\frac{(\frac{R_{mk}}{R_n^{lb}})^2-z^2}{1-z^2}}]dz(R_n^{lb})^2
\exp\{-\pi\sum\limits_{i=1}^N(\lambda_i(R_i^{lb})^2)\}dR_{mk}\\
\mathop =\limits^{(d)}8\lambda_n\lambda_mv\int\limits_0^{+\infty}\int\limits_0^{\min(1,\frac{R_{m,n}^{lb}}{x})}
[\sqrt{\frac{1-z^2}{(\frac{R_{m,n}^{lb}}{x})^2-z^2}}+\sqrt{\frac{(\frac{R_{m,n}^{lb}}{x})^2-z^2}{1-z^2}}]dzx^{\frac{\alpha_n}{\alpha_m}+1}(\frac{P_mB_m}{P_nB_n})^{\frac{1}{\alpha_m}}
\exp\{-\pi\sum\limits_{i=1}^N(\lambda_i(R_{i,n}^{lb})^2)\}dx\\
\mathop =\limits^{(e)}8\lambda_n\lambda_mv\int\limits_0^{+\infty}\int\limits_0^{\min(1,\frac{x}{R_{m,n}^{lb}})}
[\sqrt{\frac{(\frac{x}{R_{m,n}^{lb}})^2-z_1^2}{1-z_1^2}} + \sqrt{\frac{1-z_1^2}{(\frac{x}{R_{m,n}^{lb}})^2-z_1^2}}]dz_1(R_{m,n}^{lb})^2
\exp\{-\pi\sum\limits_{i=1}^N(\lambda_i(R_{i,n}^{lb})^2)\}dx\\
\mathop =\limits^{z=z_1}8\lambda_n\lambda_mv\int\limits_0^{+\infty}\int\limits_0^{\min(1,\frac{x}{R_{m,n}^{lb}})}
[\sqrt{\frac{1-z^2}{(\frac{x}{R_{m,n}^{lb}})^2-z^2}}+\sqrt{\frac{(\frac{x}{R_{m,n}^{lb}})^2-z^2}{1-z^2}}]dz(R_{m,n}^{lb})^2
\exp\{-\pi\sum\limits_{i=1}^N(\lambda_i(R_{i,n}^{lb})^2)\}dx\\
\mathop=\limits^{(f)}H_k^{n-m}
\end{array}
\end{equation}


{\ifthenelse{\boolean{publ}}{\footnotesize}{\small}
 \bibliographystyle{bmc_article}  
  \bibliography{123} }     


\ifthenelse{\boolean{publ}}{\end{multicols}}{}

\section*{Figures}
  \subsection*{Figure 1 - Example of donwlink HCNs with three tiers of BSs: high-power macrocell BSs (red square) are overlaid with successively denser and lower power picocells (red triangle) and femtocells (red circle).}

  \subsection*{Figure 2 - The admission state and its transition of the typical UE.}

  \subsection*{Figure 3 - The bad region when the typical UE moves to (r,0).}

  \subsection*{Figure 4 - The average handover rates between tiers in 1$km^2$ region with different $\overline{v}$.}

  \subsection*{Figure 5 - The total handover rate ($\lambda_h$) and handover rate between different tiers ($\lambda_h^{1-2}$) in 1$km^2$ region with different average velocity $\overline{v}$ and 2nd-tier BSs density $\lambda_2$.}

  \subsection*{Figure 6 - The handover rates between tiers in 1$km^2$ region with different average velocity $\overline{v}$ and 2nd-tier path loss factor $\alpha_2$.}

  \subsection*{Figure 7 - The forward and reverse handover rates between the two tiers in 1$km^2$ region with different average velocity $\overline{v}$, 2nd-tier BSs density $\lambda_2$ and 2nd-tier path loss factor $\alpha_2$.}

  \subsection*{Figure 8 -The CDF of residual time in the 2-tier BSs with different $\lambda_2$.}

  \subsection*{Figure 9 -The average handover arrival rates in 1$km^2$ region with different 2nd-tier bias factor ($B_2$).}

\end{bmcformat}
\end{document}